\documentclass{amsart}
\usepackage{amsmath,amssymb}
\usepackage{amsfonts}
\usepackage{amsthm}
\usepackage{mathrsfs}
\numberwithin{equation}{section}
\usepackage{comment}
\newcommand{\jap}[1]{\langle #1 \rangle}

\def\a{\alpha}
\def\b{\beta}

\def\d{\delta}
\def\e{\varepsilon}
\def\f{\varphi}
\def\g{\psi}

\def\k{\kappa}
\def\l{\lambda}
\def\m{\mu}

\def\o{\omega}

\def\s{\sigma}

\def\x{\xi}
\def\y{\eta}

\def\H{\mathcal{H}}
\def\Fd{\mathcal{F}_d}

\renewcommand{\O}{\Omega}

\def\re{\mathbb{R}}

\def\ze{\mathbb{Z}}

\def\T{\mathbb{T}}
\def\pa{\partial}

\renewcommand{\Im}{\text{{\rm Im}\;}}

\newcommand{\supp}{\text{{\rm supp}\;}}

\newtheorem{thm}{Theorem}[section]
\newtheorem{lem}[thm]{Lemma}
\newtheorem{prop}[thm]{Proposition}
\newtheorem{cor}[thm]{Corollary}

\theoremstyle{definition}
\newtheorem{defn}{Definition}

\theoremstyle{remark}
\newtheorem{rem}[thm]{Remark}

\title[]%
{Some properties of threshold eigenstates and resonant states of discrete Schr\"odinger operators}

\author{Yuji nomura}
\address{Graduate School of Material Science, University of Hyogo, Shosha, Himeji 671-2280, Japan}
\email{y.nomura@sci.u-hyogo.ac.jp}

\author{Kouichi Taira}
\address{Graduate School of Mathematical Sciences, University of Tokyo, 3-8-1 Komaba, Meguroku, Tokyo, Japan 153-8914}
\email{taira@ms.u-tokyo.ac.jp}


\begin{document}

\maketitle

\begin{abstract}
In this note, we study some properties of threshold resonant states or threshold eigenfunctions for discrete Schr\"odinger operators. We mainly prove two theorems. First, we prove that resonant states at the elliptic threshold have the same asymptotic expansion as the continuous Schr\"odinger operator. Second, we prove absence of resonant states at hyperbolic thresholds.
\end{abstract}

\section{Introduction}

We consider the discrete Schr\"odinger operators:
\begin{align*}
H=H_0+ V(x)\quad \text{on}\quad \H= l^2(\mathbb{Z}^d),
\end{align*}
where $H_0$ is the negative discrete Laplacian
\begin{align*}
H_0 u(x)=-\sum_{|x-y|=1}(u(y)-u(x)),
\end{align*}
and $V$ is a real-valued function on $\ze^d$.  
We denote the Fourier expansion by $\Fd$:
\begin{align*}
\hat{u}(\x)=\mathcal{F}_d u(\x)=\sum_{x\in \ze^d}e^{-2\pi i x\cdot\x}u(x),\quad \x\in \mathbb{T}^d=\re^d/\ze^d.
\end{align*}
Then it follows that
\begin{align}\label{mul}
\Fd H_0 u(\x)=h_0(\x)\Fd u(\x)\,\, \text{for}\,\, u\in \bigcup_{s\in \re}l^{2,s}(\ze^d)
\end{align} 
in the distributional sense, where $h_0(\x)=4\sum_{j=1}^d\sin^{2}(\pi \x_j)$, and hence $\s(H_0)=[0,4d]$. In this note, we often use $[-\frac{1}{2}, \frac{1}{2}]^d$ as a fundamental domain of $\T^d$. Moreover, we identify the integral over $\T^d$ with the integral over this fundamental domain $[-\frac{1}{2}, \frac{1}{2}]^d$. We denote $\jap{x}=(1+|x|^2)^{1/2}$ and $l^{2,s}(\ze^d)=\jap{x}^{-s}l^2(\ze^d)$. It is known that $l^{2,s}(\ze^d)$ is isometric to the Sobolev space $H^s(\T^d)$ through the Fourier expansion $\Fd$.

Critical values of $h_0$ are called thresholds of $H_0$. We denote the set of all thresholds by $\Gamma$:
\begin{align*}
\Gamma=\{\l\in [0,4d]\mid \l\,\,\text{is a critical value of}\,\, h_0\}=\{4k\}_{k=0}^d.
\end{align*}
Note that any critical points $h_0$ is non-degenerate, that is, $h_0$ is Morse.  We say that $0$ and $4d$ are elliptic thresholds and $\l\in \{4k\}_{k=1}^{d-1}$ are hyperbolic thresholds. Near each critical point of $h_0$, we have the following Taylor expansion:
\begin{align*}
h_0(\x)-\l\sim 4\pi^2(-\sum_{j=1}^{k}(\x_{\s(j)}-\y_{\s(j)})^2 +\sum_{j=k+1}^d(\x_{\s(j)}-\y_{\s(j)})^2),
\end{align*}
where $\y\in h_0^{-1}(\{\l\})$, $\l\in \Gamma,$ $k=k(\y)$ is the Morse index at $\y$ and $\s:\{1,...,d\}\to \{1,...,d\}$ is a bijection. Moreover, it easily follows that $k(\y)=0,d$ if $\y\in \{0,4d\}$ and $k(\y)\neq 0, d$ if $\y\in \Gamma\setminus\{0,4d\}$. This implies that $h_0$ behaves like the symbol $\pm |\x|^2$ of the elliptic operator $\mp \Delta$ near critical points with the elliptic thresholds and behaves like the symbol $-|\x'|^2+|\x''|^2$ ($\x=(\x',\x'')$) of the ultrahyperbolic operator $\Delta_{x'}-\Delta_{x''}$ near critical points with the hyperbolic thresholds. 

It is known that the behavior of the resolvent at thresholds is closely related to a time decay of the propagator and  that existence of eigenstates and resonant states disturbs a decay property of the propagator \cite{JK}.
Ito and Jensen obtain an analytic continuation near thresholds of the integral kernels for discrete Schr\"odinger operators \cite{IJ}. The purpose of this note is to study some properties of resonant states: Resonant states at elliptic thresholds have same properties as continuous one's and resonances at hyperbolic thresholds disappear. From this, we expect that the hyperbolic thresholds is harmless for a decay property of the propagator.

First, we give a definition of resonances at elliptic thresholds.

\begin{defn}
Let $d\geq 3$ and $\l=0$ or $4d$. Suppose that a real-valued function $V$ satisfies $|V|\leq C\jap{x}^{-2-\d}$ with $\d>0$. We say that $u\in l^{2,-3/2}(\ze^d)\setminus l^2(\ze^d)$ is a resonant state of $H=H_0+V$ if $u$ satisfies
\begin{align*}
Hu=\l u.
\end{align*}
If such $u$ exists, we say that $\l$ is a resonance of $H$.
\end{defn}

From now on, we concentrate the case of $\l=0$. Now we state our first theorem, which is an analogy of the continuous model (for example, see \cite[Lemma 2.4]{Y}).

\begin{thm}\label{mainprop}
Let $d\geq 3$.  Suppose $V$ be a real-valued function satisfying $|V|\leq C\jap{x}^{-2-\e}$ for $0<\e\leq 1$ and $u\in l^{2, -3/2}(\ze^d)$ satisfies $(H_0+V)u=0$. Then there exists $C>0$ such that
\begin{align*}
&|u(x)|\leq C\jap{x}^{-d+2},\\
&u(x)=-c_d|x|^{-d+2}\sum_{y\in\ze^d}Vu(y)+O(|x|^{-d+2-\e})
\end{align*}
as $|x|\to \infty$, where
\begin{align}\label{c_d}
c_d=\frac{\Gamma(\frac{d}{2}-1)}{4\pi^{\frac{d}{2}}}.
\end{align}
In particular, if $\sum_{x\in \ze^d}Vu(x)\neq 0$ holds, then  $|u(x)|\geq C|x|^{-d+2}$ follows as $|x|\to \infty$. 
\end{thm}

\begin{rem}
This theorem implies that 

\item[$(i)$] Set $N_s=\{u\in l^{2,-s}(\ze^d)\mid (H_0+V)u=0\}$ for $1/2<s\leq 3/2$. Then $N_s=N_{s'}$ for $s,s'\in (1/2,3/2]$.

\item[$(ii)$] Suppose that $d=3$ with $\e>1/2$ or $d=4$ with $\e>0$. Then it follows that the above $u$ is an $l^2$-eigenfunction of $H_0+V$ if and only if $\sum_{y\in \ze^d}Vu(y)=0$.

\item[$(iii)$] There are no resonances at zero energy for $d\geq 5$.
\end{rem}

Let $d\geq 3$. We recall some results from \cite[Theorem 1.1, Theorem 1.8 and Proposition 3.4]{TT}. We have the following limiting absorption principle with the thresholds weight:
\begin{align}\label{LAP1}
\sup_{z\in \mathbb{C}\setminus \re}\|\jap{x}^{-1+\d}(H_0-z)^{-1}\jap{x}^{-1-\d}\|_{B(l^2(\ze^d))}<\infty
\end{align}
if $|\d|\geq 0$ is small enough.
Moreover, the following limits exist in $B(l^{2, s}(\ze^d) , l^{2,-s}(\ze^d))$ for $s>1$:
\begin{align}\label{LAP2}
(H_0-\lambda \mp i0)^{-1}:= \lim_{\e\to 0, \, \e>0}(H_0-\lambda \mp i\e)^{-1},\,\, \l\in [0,4d].
\end{align}
We note that $(\ref{LAP1})$ and $(\ref{LAP2})$ away from $\Gamma$ directly follow from the Mourre theory or \cite[Proposition B.5]{TT}. The novelty of $(\ref{LAP1})$ and $(\ref{LAP2})$ lie in the estimates near $z, \l\in \Gamma$. Furthermore, we have the following lemma which immediately follows from a density argument.
\begin{lem}
Let $d\geq 3$.
The operators $(H_0-\lambda \mp i0)^{-1}\in B(l^{2, s}(\ze^d) , l^{2,-s}(\ze^d))$ for $s>1$ and $\l\in [0,4d]$ uniquely extend to bounded linear operators from $l^{2,1}(\ze^d)$ to $l^{2,-1}(\ze^d)$. Moreover, we have
\begin{align}\label{LAP3}
\sup_{\l\in \re}\|\jap{x}^{-1}(H_0-\l\mp i0)^{-1}\jap{x}^{-1}\|_{B(l^2(\ze^d))}<\infty.
\end{align}
\end{lem}

\begin{rem}
This lemma does not assert 
\begin{align*}
(H_0-\lambda \mp i0)^{-1}=\lim_{\e\to 0,\e>0}(H_0-\lambda \mp i\e)^{-1}\,\, \text{in}\,\, B(l^{2,1}(\ze^d), l^{2,-1}(\ze^d)).
\end{align*}

\end{rem}

Now we give a definition of resonance at hyperbolic thresholds.

\begin{defn}
Let $d\geq 3$.
Suppose that a real-valued function $V$ satisfies $|V|\leq C\jap{x}^{-2-\d}$ with $\d>0$.
Let $\l\in \Gamma\setminus \{0,4d\}$, that is, $\l$ is a hyperbolic threshold. We call $u\in l^{2,-1}(\ze^d)\setminus l^2(\ze^d)$ is a resonant state of $H=H_0+V$ if $u$ satisfies
\begin{align*}
u+(H_0-\l\mp i0)^{-1}Vu=0.
\end{align*}
If such $u$ exists, we call that $\l$ is a resonance of $H$.

\end{defn}
\begin{rem}
The validity of this definition lies in  Proposition \ref{lappropp}: If $\l$ is not an eigenvalue and not a resonance of $H$, then the outgoing/ incoming resolvent $(H-\l\mp i0)^{-1}$ exist.
\end{rem}

\begin{rem}
As is shown in Lemma \ref{hypreg}, we can replace $u\in l^{2,-1}(\ze^d)$ by $l^{2,-1-\d}(\ze^d)$.
\end{rem}

The following theorem implies that resonances of $H$ at hyperbolic thresholds disappear under a stronger assumption of $V$ even when $d=3$ or $4$.

\begin{thm}\label{hypthm}
Let $d\geq 3$ and $V$ be a real-valued function satisfying $|V(x)|\leq C\jap{x}^{-\d}$ with $\d>d/2+2$ . If $u\in l^{2,-1}(\ze^d)$ satisfies $u+(H_0-\l\pm i0)Vu=0$, then $u\in l^2(\ze^n)$. 
\end{thm}

We recall from \cite{IM} that for a finitely supported real-valued potential $V$, $H$ has no eigenvalues in $(0,4d)$. Combining this result with Theorem \ref{hypthm}, we obtain the following corollary.

\begin{cor}\label{corab}
Let $d\geq 3$ and $V$ be a finitely supported real-valued potential. Then $H_0+V$ has no resonances and no eigenvalues in $(0,4d)$.
\end{cor}

This corollary implies the limiting absorption principle for $H=H_0+V$ near hyperbolic thresholds.

\begin{thm}\label{lapthm}
Let $d\geq 3$ and $V$ be a finitely supported real-valued potential. Set 
\begin{align*}
\O_{\e_1,\pm }=\{z\in \mathbb{C}\mid \pm\Im z>0 ,\,\,|z|>\e_1,\,\, |z-4d|>\e_1\}.
\end{align*}
for $0<\e_1<1$ and a signature $\pm$. Now fix a constant $0<\e_1<1$ and a signature $\pm$.

\item[$(i)$] We have
\begin{align}\label{LAP4}
\sup_{z\in \O_{\e_1, \pm}}\|\jap{x}^{-1}(H-z)^{-1}\jap{x}^{-1}\|_{B(l^2(\ze^d))}<\infty.
\end{align}
\item[$(ii)$] For each $s>1$, the operators $z\in \O_{\e_1,\pm} \mapsto (H-z)^{-1}\in B(l^{2,s}(\ze^d), l^{2,-s}(\ze^d))$ is H\"older continuous. In particular, limits
\begin{align*}
(H-\l\mp i0)^{-1}:=\lim_{\e\to 0,\, \e>0} (H-\l\mp i\e)^{-1}
\end{align*}
exist in the norm operator topology of $B(l^{2,s}(\ze^d), l^{2,-s}(\ze^d))$ for $\e_1<\l<4d-\e_1$.

\item[$(iii)$] Let $s>1$ and $\e_1<\l<4d-\e_1$. The outgoing/incoming resolvents $(H-\lambda \mp i0)^{-1}\in B(l^{2, s}(\ze^d) , l^{2,-s}(\ze^d))$ uniquely extend to bounded linear operators from $l^{2,1}(\ze^d)$ to $l^{2,-1}(\ze^d)$. Moreover, we have
\begin{align}\label{LAP5}
\sup_{\e_1<\l<4d-\e_1}\|\jap{x}^{-1}(H-\l\mp i0)^{-1}\jap{x}^{-1}\|_{B(l^2(\ze^d))}<\infty.
\end{align}
\end{thm}

\begin{rem}
Suppose that there are no resonances and no eigenvalues at $\{0,4d\}$. Then the all results in the above theorem still hold if we replace $\O_{\e_1, \pm}$ by $\mathbb{C}_{\pm}=\{z\in \mathbb{C}\mid \pm\Im z>0\}$. See Proposition \ref{lappropp}.
\end{rem}

As mentioned above, for the case of finitely supported potentials it is known that there are no eigenvalues in open interval $(0,  4d)$ (see [3]).
However it is possible that the threshold $0$ or $4d$ might be embedded
eigenvalue.
The persistent set (variety) $P_S$ of embedded eigenvalue $0$ is defined as the set of all potentials $V$  supported on $S$ such that $H=H_0 + V$ has the
eigenvalue $0$,
that is
\[
P_{S} = \{ V \subset  {\mathbb{R}}^S \mid \mathrm{supp}  V \subset S\,\, \text{and}\,\,  0 \text{ is an
eigenvalue of}\,\, H_{0} + V\}.
\]
Here  $S$ is a fixed finite subset of   $\mathbb{Z}^d$.
In [HNO], some geometrical structure and properties of $P_{S}$ are
considered.
Moreover the notion of the threshold resonances is defined and
non-existence of them
for $d \ge 5$ and the persistent set of them for $d=2,3,4$ are studied.
The ways of proofs for many statements in [HNO], however, seem to depend on
the finiteness
of potential support.
So in our article we attempt to give an appropriate definition of threshold
resonat states of  more general potentials and investigate some properties
of them
 by using a method of  harmonic analysis. Furthermore we study the limiting absorption principle and resonances at hyperbolic thresholds.

We fix some notations. For Banach spaces $X,Y$, we denote the set of all bounded linear operators from $X$ to $Y$ by $B(X,Y)$ and set $B(X):=B(X,X)$. 

We need the following useful representation.
We assume $\nabla h_0\neq 0$ on $\{h_0(\x)=\l\}\cap U$ for $\l\in\re$ and and an open set $U$. Moreover, we assume $\{h_0(\x)=\l\}\cap U$ has the  following graph representation: 
\begin{align*}
\{h_0(\x)=\l\}\cap U=\{\x\mid \x_d=g(\x')\},\,\, \x=(\x',\x_d).
\end{align*}
Then the induced surface measure $d\s$ on $\{h_0(\x)=\l\}\cap U$ is written as
\begin{align}\label{surfrep}
d\s(\x)=\sqrt{1+|\nabla g(\x')|^2}d\x'=\frac{|(\nabla_{\x}h_0)(\x', g(\x'))|}{|(\pa_{\x_d}h_0)(\x', g(\x'))|}d\x'.
\end{align}

\textbf{Acknowledgment.} 
YN was supported by JSPS KAKENHI Grant number  5K04960. KT was supported by JSPS Research Fellowship for Young Scientists, KAKENHI Grant Number 17J04478 and the program FMSP at the Graduate School of Mathematics Sciences, the University of Tokyo. KT would like to thank his supervisors Kenichi Ito and Shu Nakamura for encouraging to write this paper.

\section{Pointwise estimates, Proof of Theorem \ref{mainprop}}

\subsection{Upper bounds}
Let $d\geq 3$. We consider the solution to 
\begin{align}\label{Scheq}
(H_0+V)u=0.
\end{align}
First, we reduce the equation $(\ref{Scheq})$ to the integral equation, which is useful for estimating $u$:
\begin{align}\label{eig}
u+H_0^{-1}Vu=0,
\end{align}
where
\begin{align*}K_2(x)=\int_{\T^d}e^{2\pi x\cdot \x}h_0(\x)^{-1}d\x,\,\,  H_0^{-1}w(x)=\sum_{y\in \mathbb{Z}^d}K_2(x-y)w(y),
\end{align*}
for $w\in l^{2, 1/2+\e}(\ze^d)$ with $\e>0$. Here $H_0^{-1}$ is the bounded operator from $l^{2,\a}(\ze^d)$ to $l^{2,-\b}(\ze^d)$ for $\a,\b>1/2$ with $\a+\b\geq 2$ (see Appendix B, Corollary \ref{disHLS}). Moreover, it also follows that the multiplication operator 
\begin{align*}
h_0^{-1}:\bigcap_{s>0}l^{2,s}(\ze^d)\to \bigcup_{s\in \re}l^{2,s}(\ze^d)
\end{align*}
can be uniquely extended to the operator
\begin{align}\label{h_0map}
h_0^{-1}:H^{\a}(\T^d)\to H^{-\b}(\T^d),\,\, \a,\b>\frac{1}{2},\,\, \a+\b\geq 2
\end{align}
and that
\begin{align*}
h_0^{-1}=\mathcal{F}_d^{-1}H_0^{-1}\mathcal{F}_d   :H^{\a}(\T^d)\to H^{-\b}(\T^d),\,\, \a,\b>\frac{1}{2},\,\, \a+\b\geq 2.
\end{align*}

\begin{lem}
We assume $|V(x)|\leq C\jap{x}^{-2-\e}$ for some $\e>0$. 
For $u\in l^{2,-3/2}(\ze^d)$, $(\ref{Scheq})$ implies $(\ref{eig})$.

\end{lem}
\begin{proof}
The relations $(\ref{mul})$ and $(\ref{Scheq})$ implies
\begin{align}\label{ei1}
h_0(\x)\hat{u}(\x)=-\widehat{Vu}(\x),\,\, \hat{u}\in H^{-\s}(\T^d).
\end{align}
First, we note $\hat{u}(\x)=-h_0(\x)^{-1}\widehat{Vu}(\x)$ in $\mathcal{D}'(\T^d\setminus \{0\})$. To see this, it suffices to prove
\begin{align*}
(\hat{u},\f)=-(h_0^{-1}\widehat{Vu}, \f),\,\, \f\in C_c^{\infty}(\T^d\setminus \{0\}).
\end{align*}
Let $\g\in C_c^{\infty}(\T^d\setminus \{0\})$ be a real-valued function such that $\g\f=\f$. Then we have
\begin{align}\label{ei2}
\g(\x)\hat{u}(\x)=-\g(\x)h_0(\x)^{-1}\widehat{Vu}(\x)
\end{align}
in the distributional sense.
In fact, since $\widehat{Vu}\in L^2(\T^d)$ and since $h_0(\x)^{-1}$ is smooth away from $\x=0$, then it follows that the both side of $(\ref{ei1})$ are measurable functions away from $\x=0$. In particular, $\g(\x)h_0(\x)\hat{u}(\x)$ and $\g(\x)\widehat{Vu}(\x)$ are measurable functions. This implies $(\ref{ei2})$ as measurable functions. Since $\widehat{Vu}\in L^2(\T^d)\subset L^1(\T^d)$ and since $\g(\x)h_0(\x)^{-1}$ is smooth, then it follows that $\g(\x)\hat{u}\in L^2(\T^d)\subset L^1(\T^d)$. Thus $(\ref{ei2})$ follows in $L^1(\T^d)$. In particular, we have $(\ref{ei2})$ in the distributional sense. Hence we obtain
\begin{align*}
(\hat{u},\f)=(\g\hat{u},\f)=-(\g h_0^{-1}\widehat{Vu}, \f)=-(h_0^{-1}\widehat{Vu}, \f).
\end{align*}
This proves $\hat{u}(\x)=-h_0(\x)^{-1}\widehat{Vu}(\x)$ in $\mathcal{D}'(\T^d\setminus \{0\})$. We note $h_0^{-1}\widehat{Vu}\in \mathcal{D}'(\T^d)$ by $(\ref{h_0map})$. These imply that $\hat{u}+h_0^{-1}\widehat{Vu}$ is supported in $\{0\}$ as an element of $\mathcal{D}'(\T^d)$ and can be written as a linear combination of the derivatives of the Dirac measure. Since $\pa_{\x}^{\a}\d\notin H^{-d/2}(\T^d)$ for ant $\a\in \mathbb{N}^d$, it suffices to prove $\hat{u}+h_0^{-1}\widehat{Vu}\in H^{-d/2}(\T^d)$ in order to deduce $\hat{u}=-h_0^{-1}\widehat{Vu}$. Since $\hat{u}\in H^{-3/2}(\T^d)\subset H^{-d/2}(\T^d)$, we only need to prove $h_0^{-1}\widehat{Vu}\in H^{-d/2}(\T^d)$. Using $\widehat{Vu}\in H^{1/2+\e}(\T^d)$ and $(\ref{h_0map})$ with $\a=1/2+\e$ and $\b=3/2$, we obtain $h_0^{-1}\widehat{Vu}\in H^{-3/2}(\T^d)\subset H^{-d/2}(\T^d)$. This completes the proof.
\end{proof}

The main result of this subsection is the following proposition.

\begin{prop}\label{upp}
Let $u\in l^{2,-3/2}(\ze^d)$ be a solution to $(\ref{eig})$. Then we have
\begin{align*}
|u(x)|\leq C\jap{x}^{-d+2}.
\end{align*}
\end{prop}
The following lemma is useful.

\begin{lem}\label{intcal}
Let $d\geq 1$. 
\item[$(i)$]
Let $k, l<d$ with $k+l>d$. Then we have
\begin{align*}
I=\sum_{y\in \ze^d}\jap{x-y}^{-k}\jap{y}^{-l}\leq C\jap{x}^{d-k-l}.
\end{align*}

\item[$(ii)$] Let $0<k<d$ and $l=d$. For any $\d>0$, there exists $C_{\d}>0$ such that $I\leq C_{\d}\jap{x}^{\d-k}$.

\item[$(iii)$] Let $0<k<d<l$. Then we have
\begin{align*}
I\leq C\jap{x}^{-k}.
\end{align*}

\item[$(iv)$] Let $k=d$ and $l>d$. Then we have
\begin{align*}
I\leq C\jap{x}^{-d}.
\end{align*}

\end{lem}

\begin{proof}
$(i)$ We decompose $I=I_1+I_2+I_3$ such that
\begin{align*}
&I_1=\sum_{|x-y|\leq 1/2|x|}\jap{x-y}^{-k}\jap{y}^{-l},\,\, I_2=\sum_{\substack{|x-y|\geq 1/2|x|,\\ |y|\leq 2|x|}}\jap{x-y}^{-k}\jap{y}^{-l},\\
&I_3=\sum_{\substack{|x-y|\geq 1/2|x|,\\ |y|> 2|x|}}\jap{x-y}^{-k}\jap{y}^{-l}.
\end{align*}
We note that $|x-y|\leq 1/2|x|$ implies $1/2|x|\leq |y|\leq 3/2|x|$. Using this and $k<d$, we have
\begin{align*}
I_1\leq& C\jap{x}^{-l}\sum_{|x-y|\leq 1/2|x|}\jap{x-y}^{-k}=C\jap{x}^{-l}\sum_{|y|\leq 1/2|x|}\jap{y}^{-k}\leq C\jap{x}^{d-k-l}.
\end{align*}
Moreover, using $l<d$, we learn
\begin{align*}
I_2\leq C\jap{x}^{-k}\sum_{\substack{|x-y|\geq 1/2|x|,\\ |y|\leq 2|x|}}\jap{y}^{-l}\leq C\jap{x}^{d-k-l}.
\end{align*}
To estimate $I_3$, we observe that $|x-y|\geq 1/2|y|$ holds in $\{|y|>2|x|\}$. Using this and $k+l>d$, we obtain
\begin{align*}
I_3\leq C\sum_{|y|>2|x|}\jap{y}^{-k-l}\leq C\jap{x}^{d-k-l}.
\end{align*}
Thus we conclude $I\leq C\jap{x}^{d-k-l}$.

$(ii)$ As in the proof of $(i)$, using $k<d$ and $k+l>d$ with $l=d$, we have $I_1+I_3\leq C\jap{x}^{-k}$. We observe
\begin{align*}
I_2\leq C\jap{x}^{-k}\sum_{|y|\leq 2|x|}\jap{y}^{-d}\leq C_{\d}\jap{x}^{\d-k}.
\end{align*}
This proves $(ii)$.

$(iii)$ As in the proof of $(i)$, using $k<d$, we have $I_1\leq C\jap{x}^{d-k-l}$. The inequality $l>d$ implies $I_1\leq C\jap{x}^{-k}$. On the other hand, using $l>d$, we observe
\begin{align*}
I_2+I_3\leq C\jap{x}^{-k}\sum_{|y|\leq 2|x|}\jap{y}^{-l}\leq C\jap{x}^{-k}.
\end{align*}
We conclude $I\leq C\jap{x}^{-k}$.

$(iv)$ As in the proof of $(iii)$, using $l>d$, we have $I_2+I_3\leq C\jap{x}^{-d}$. Since $|x-y|\leq 1/2|x|$ holds on $\{1/2|x|\leq |y|\leq 3/2|x|\}$, we have
\begin{align*}
I_1\leq C\jap{x}^{-l}\sum_{|y|\leq 1/2|x|}\jap{x}^{-d}\leq C_{\d}\jap{x}^{\d-l}
\end{align*}
for any $\d>0$. We take $\d=l-d>0$ and obtain $I_3\leq C\jap{x}^{-d}$.

\end{proof}

\begin{proof}[Proof of Proposition \ref{upp}]
We may assume $0<\e<1$. Using $u\in l^{2,-3/2}(\ze^d)$, $|V(x)|\leq C\jap{x}^{-2-\e}$ and Corollary \ref{Kerpo} with $l=2$, we have
\begin{align*}
|u(x)|=|H_0^{-1}Vu(x)|\leq& C\sum_{y\in \ze^d}\jap{y}^{-d+2}|Vu(x-y)|\\
\leq&C(\sum_{y\in \ze^d}\jap{y}^{-2d+4}\jap{x-y}^{-1-2\e})^{1/2}\|\jap{x}^{1/2+\e}Vu\|_{l^2(\ze^d)}.
\end{align*}
Applying Lemma \ref{intcal} with $k=1+2\e$ and $l=2d-4$, we have $|u(x)|\leq C\jap{x}^{-\e}\leq C\jap{x}^{-\e/2}$ for $d=3$, $|u(x)|\leq C\jap{x}^{-1/2-\e/2}$ for $d\geq 4$. 

The argument below is based on the standard bootstrap technique (for example, see \cite[Lemma 8 in the proof of Theorem XIII.33]{RS}).
Set $\a_d=0$ for $d=3$ and $\a_d=1/2$ for $d\geq 4$. Let $N$ be a real number such that $2+\a_d+(N+1)\e<d$. Suppose $|u(x)|\leq C\jap{x}^{-\a_d-N\e}$ holds. Then it follows that
\begin{align*}
|u(x)|\leq C\sum_{y\in \ze^d}\jap{y}^{-d+2}\jap{x-y}^{-2-\a_d-(N+1)\e}
\end{align*}
Applying Lemma \ref{intcal} with $k=2+\a_d+(N+1)\e$ and $l=d-2$, we have $|u(x)|\leq C\jap{x}^{-\a_d-(N+1)\e}$. By an induction argument, we obtain $|u(x)|\leq C\jap{x}^{-d+2}$.

\end{proof}

\subsection{Lower bounds, Proof of Theorem \ref{mainprop}}

We need some elementary lemmas.

\begin{lem}\label{foucon}\cite[Theorem2.4.6]{G}
Let $c_d>0$ be as in $(\ref{c_d})$. Then we have
\begin{align*}
\int_{\re^d}e^{2\pi ix\cdot \x}\frac{1}{4\pi^2|\x|^2}d\x=c_d|x|^{-d+2}.
\end{align*}
\end{lem}
We omit the proof of this lemma.

\begin{lem}\label{diffes}
There exists $C>0$ such that 
\begin{align*}
||x-y|^{-d+2}-|x|^{-d+2}|\leq C|x|^{-d+1}|y|
\end{align*}
for $x,y\in \re^d$ with $|x|/2>|y|$.
\end{lem}
\begin{proof}
First, we assume $|x-y|\leq |x|$.
For $|x|/2>|y|$, we have
\begin{align}
\frac{1}{|x-y|^{d-2}}=&\frac{1}{|x|^{d-2}|\frac{x}{|x|}-\frac{y}{|x|}|^{d-2}}\leq \frac{1}{|x|^{d-2}(1-\frac{|y|}{|x|})^{d-2}}\label{Ele}\\
=&|x|^{-d+2}(1+\frac{|y|}{|x|}\sum_{j=0}^{\infty}\frac{|y|^j}{|x|^j})^{d-2}\leq |x|^{-d+2}(1+2\frac{|y|}{|x|})^{d-2}\nonumber\\
\leq&|x|^{-d+2}+C|x|^{-d+1}|y|.
\end{align}

Next, we assume $|x-y|>|x|$. Setting $x'=x-y$ and $y'=-y$, we have $|x'-y'|\leq |x'|$ and $|x'|/2>|y'|$. Applying $(\ref{Ele})$ with $x=x'$ and $y'=y$, we obtain
\begin{align*}
|x|^{-d+2}-|x-y|^{-d+2}\leq C|x|^{-d+1}|y|.
\end{align*}
Thus we obtain $||x-y|^{-d+2}-|x|^{-d+2} |\leq C|y||x|^{-d+1}$. 
\end{proof}

\begin{proof}[Proof of Theorem \ref{mainprop}]
Note that $|Vu(x)|\leq C\jap{x}^{-d-\e}$ and
\begin{align*}
u(x)=-\sum_{y\in \ze^d}G(x,y)Vu(y),\,\, G(x,y)=\int_{\T^d}e^{2\pi i(x-y)\cdot \x}\frac{1}{h_0(\x)}d\x.
\end{align*}
For small $r>0$, take $\chi\in C^{\infty}(\T^d,[0,1])$ such that $\chi=1$ on $|\x|\leq r$ and $\chi=0$ outside $|\x|\leq 2r$. Then
\begin{align*}
u(x)=-\sum_{y\in \ze^d}G_1(x,y)Vu(y)+O(\jap{x}^{-\infty}),\,\, G_1(x,y)=\int_{\T^d}e^{2\pi i(x-y)\cdot \x}\frac{\chi(\x)}{h_0(\x)}d\x.
\end{align*}
We use the following lemmas.
\begin{lem}\label{lem1}
We have
\begin{align*}
u(x)=-\sum_{y\in \ze^d}G_2(x,y)Vu(y)+O(\jap{x}^{-d}),
\end{align*}
where $G_2(x,y)=\int_{\re^d}e^{2\pi i(x-y)\cdot \x}\frac{\chi(\x)}{4\pi^2 |\x|^2}d\x$.
\end{lem}

\begin{proof}

If $|\x|\leq 2r$ for small $r>0$, then we expand $h_0(\x)^{-1}=1/(4\pi|\x|^2)+R(\x)$, where $|\pa_{\x}^{\a}R(\x)|\leq C_{\a}|\x|^{-|\a|}$. Thus we have
\begin{align*}
u(x)=-\sum_{y\in \ze^d}G_2(x,y)Vu(y)-\sum_{y\in \ze^d}G_3(x,y)Vu(y)+O(\jap{x}^{-\infty}),
\end{align*}
where 
\begin{align*}
G_3(x,y)=\int_{\re^3}e^{2\pi i(x-y)\cdot \x}\chi(\x)R(\x)d\x.
\end{align*}
By Lemmas \ref{mlem} and \ref{intcal} $(iv)$ with $k=d$ and $l=d+\e$, the second term is $O(|x|^{-d})$. This completes the proof.
\end{proof}

\begin{lem}\label{asym}
For $|x|\geq 1$, we have
\begin{align*}
\int_{\re^d}e^{2\pi ix\cdot \x}\frac{\chi(\x)}{4\pi^2|\x|^2}d\x=c_d|x|^{-d+2}+O(\jap{x}^{-d+1}).
\end{align*}
\end{lem}

\begin{proof}

By Lemma \ref{foucon} and the Plancherel theorem, we notice that 
\begin{align*}
\int_{\re^d}e^{2\pi ix\cdot \x}\frac{\chi(\x)}{4\pi^2|\x|^2}d\x=c_d\int_{\re^d}\frac{1}{|x-y|^{d-2}}\mathcal{F}^{-1}\chi(y)dy.
\end{align*}
Since $\mathcal{F}^{-1}\chi$ is rapidly decreasing, we have

\begin{align*}
\int_{|x|/2<|y|}\frac{1}{|x-y|^{d-2}}\mathcal{F}^{-1}\chi(y)dy)=O(\jap{x}^{-\infty}).
\end{align*}
Thus it suffices to prove that
\begin{align}\label{ma}
c_d\int_{|x|/2>|y|}\frac{1}{|x-y|^{d-2}}\mathcal{F}^{-1}\chi(y)dy=c_d|x|^{-d+2}+O(\jap{x}^{-d+1}).
\end{align}
By Lemma \ref{diffes}, we have
\begin{align*}
c_d\int_{|x|/2>|y|}\frac{1}{|x-y|^{d-2}}\mathcal{F}^{-1}\chi(y)dy=&c_d|x|^{-d+2}\int_{|x|/2>|y|}\mathcal{F}^{-1}\chi(y)dy+O(\jap{x}^{-d+1})\\
=&c_d|x|^{-d+2}\int_{\re^d}\mathcal{F}^{-1}\chi(y)dy+O(\jap{x}^{-d+1})\\
=&c_d\chi(0)|x|^{-d+2}+O(\jap{x}^{-d+1}),
\end{align*}
where we use the fact that $\mathcal{F}^{-1}\chi$ is rapidly decreasing in the second line.
Using $\int_{\re^d}\mathcal{F}^{-1}\chi(y)dy=\chi(0)=1$, we obtain $(\ref{ma})$. 
\end{proof}

\begin{lem}\label{lem3}
\begin{align*}
\sum_{|y|\geq  |x|/2}G_2(x,y)Vu(y)=O(\jap{x}^{-d+2-\e}).
\end{align*}
\end{lem}
\begin{proof}
By Lemma \ref{mlem}, we have $G_2(x,y)=O(\jap{x-y}^{-d+2})$. Since $V(x)=O(\jap{x}^{-2-\e})$ holds, by Proposition \ref{upp}, we have $Vu=O(\jap{x}^{-d-\e})$. Now the lemma is proved by an easy calculation using the condition $\{|y|\geq |x|/ 2\}$.
\end{proof}

\begin{lem}\label{lem4}
\begin{align*}
\sum_{|y|<1/2|x|}G_2(x,y)Vu(y)=c_d|x|^{-d+2}\sum_{y\in \ze^d}Vu(y)+O(\jap{x}^{-d+1-\e}).
\end{align*}
\end{lem}
\begin{proof}
By Lemma \ref{asym} and $Vu=O(\jap{x}^{-d-\e})$, we have
\begin{align*}
\sum_{|y|<1/2|x|}G_2(x,y)Vu(y)=&c_d\sum_{|y|<1/2|x|}|x-y|^{-d+2}Vu(y)\\
&+\sum_{|y|<1/2|x|}O(\jap{x-y}^{-d+1}\jap{y}^{-d-\e})\\
=&c_d\sum_{|y|<1/2|x|}|x-y|^{-d+2}Vu(y)+O(\jap{x}^{-d+1}),
\end{align*}
where we use Lemma \ref{intcal} with $k=-d+1$ and $l=-d-\e$ in the second line. By Lemma \ref{diffes}, we have
\begin{align*}
c_d\sum_{|y|<1/2|x|}|x-y|^{-d+2}Vu(y)=&c_d|x|^{-d+2}\sum_{|y|<1/2|x|}Vu(y)\\
&+\sum_{|y|<1/2|x|}O(\jap{x}^{-d+1}\jap{y}^{-d+1-\e})\\
=&c_d|x|^{-d+2}\sum_{|y|<1/2|x|}Vu(y)+O(\jap{x}^{-d+2-\e})\\
=&c_d|x|^{-d+2}\sum_{y\in \ze^d}Vu(y)+c_d|x|^{-d+2}\sum_{|y|\geq 1/2|x|}Vu(y)\\
&+O(\jap{x}^{-d+2-\e})\\
=&c_d|x|^{-d+2}\sum_{y\in \ze^d}Vu(y)+O(\jap{x}^{-d+2-\e}),
\end{align*}
where we use $Vu=O(\jap{x}^{-d-\e})$. This completes the proof.
\end{proof}

We return to the proof of Theorem \ref{mainprop}. 
By virtue of Lemmas \ref{lem1} and \ref{lem3}, we write
\begin{align*}
u(x)=-\sum_{|y|<1/2|x|}G_2(x,y)Vu(y)+O(|x|^{-d+2-\e}).
\end{align*}
Note that $|x-y|$ is large if $|x|$ is large and $|y|<1/2|x|$. Using Lemma \ref{lem4}, we complete the proof of Theorem \ref{mainprop}.

\end{proof}

\section{Absence of embedded resonances, Proof of Theorem \ref{hypthm}}

\subsection{Preliminary lemmas}

Let $d\geq 3$ and $\l\in \{4k\}_{k=1}^{d-1}$. Set $M_{\l}=\{\x\in \T^d\mid h_0(\x)=\l\}$ and 
\begin{align*}
\Sigma_{\l}=&\{\x\in \T^d\mid \nabla h_0(\x)=0\}=\{\x\in \T^d\mid \sin 2\pi \x_j=0,\,\, \text{for all}\,\, j=1,...,d\}\\
=&\{\x\in \T^d\mid \x_j\in \{0,\frac{1}{2}\},\,\, \text{for all}\,\, j=1,...,d\}.
\end{align*}
We note $M_{\l}\setminus \Sigma_{\l}$ is an embedded submanifold of $\T^d$ with codimension $1$ and $M_{\l}$ is a Lipschitz submanifold in the sense that $M_{\l}$ has a graph representation by a Lipschitz function.
We denote the induced surface measure of $M_{\l}$ by $d\s(\x)$. Set
\begin{align*}
d\m(\x)=\frac{1}{|\nabla h_0(\x)|}d\s(\x)
\end{align*}
We note that $|\nabla h_0(\x)|^{-1}\sim |\x|^{-1}$ near $\Sigma_{\l}$ implies that $d\m$ is singular measure on $M_{\l}$ for $\l\in \Gamma$, though $|\nabla h_0(\x)|^{-1}$ is harmless on $M_{\l}$ with a regular value $\l$.
Moreover, we denote $R_0(\l\pm i0)=(H_0-\l\mp i0)^{-1}\in B(l^{2,1}(\ze^d), l^{2,-1}(\ze^d))$. First, we show $\Sigma_{\l}$ is measure zero with respect to $d\s$ and $d\m$, which essentially follows from the fact that $d\s$ and $d\m$ are finite sums of the absolutely continuous measures with respect to $d-1$-dimensional Lebesgue measure.

\begin{lem}
$d\s(\Sigma_{\l})=0$ and $d\m(\Sigma_{\l})=0$.
\end{lem}
\begin{proof}
First, we note that the measure $\m$ is absolutely continuous with respect to $d\s$. To see this, it suffices to show that $1/|\nabla h_0(\x)|$ is integrable with respect to the measure $\s$. We note that for $\y\in \Sigma_{\l}$ and $\x=(\x',\x_d)\in M_{\l}$, we have $|\nabla h_0(\x)|\sim 2\pi|\x-\y|\sim C|\x'-\y'|$ near $\x=\y$ and $\pm (\x_d-\y_d)\geq  |\x'-\y'|/2d$. The integrability of $1/|\x'-\y'|$ over $\{\x'\in \re^{d-1}\mid |\x'-\y'|:\text{small}\}$ which follows from the assumption $d\geq 3$, implies $1/|\nabla h_0(\x)|$ is integrable over $\{\pm (\x_d-\y_d)\geq |\x'-\y'|\}$. By using a partition of unity, the integrability of $1/|\nabla h_0(\x)|$ over $M_{\l}$ follows.

Thus a proof of $d\m(\Sigma_{\l})=0$ reduces to a proof of $d\s(\Sigma_{\l})=0$. Let $\y\in \Sigma_{\l}$. Since $\#\Sigma<\infty$, it suffices to prove that $\{\y\}$ has zero measure with respect to $\chi d\s$, where $\chi\in C^{\infty}(\T^d)$ is any function supported close to $\y$. Set 
\begin{align*}
A_{j,\pm}=\{\x\in \supp \chi\mid \pm (\x_j-\y_j)\geq |\x-\y|/2d\}.
\end{align*}
Then we have
\begin{align*}
 \chi(\x) d \s(\x)=\sum_{j=1...,d,\, a=\pm}\chi_{A_{j,a}}(\x) \chi(\x)d\s(\x)=:\sum_{j=1...,d,\, a=\pm}d\s_{j,a}(\x),
\end{align*}
where $\chi_{A}$ is the characteristic function of $A\subset \T^d$. Thus it suffices to prove that $\{\y\}$ is zero measure with respect to $d\s_{j,a}$ for any $j=1,..,d$ and $a=\pm$.

By rotating and reflecting the coordinate, we may assume $j=d$ and $a=+$. If $\supp \chi$ is small enough, we have the following graph representation:
\begin{align*}
M_{\l}\cap \supp\chi\cap \{\pm (\x_d-\y_d)\geq |\x-\y|/2d\}=\{(\x',g(\x'))\}
\end{align*}
where $g$ is a Lipschitz function. On this coordinate, we write
\begin{align*}
d \s_{d,+}(\x)=\chi_{A_{j,a}}(\x)\chi(\x)\sqrt{1+|\nabla_{\x'} g(\x')|^2}d\x'
\end{align*}
by $(\ref{surfrep})$. This implies that $d\s_{d,+}$ is absolutely continuous with respect to the $d-1$-dimensional Lebesgue measure $d\x'$. This completes the proof.

\end{proof}

We recall the standard $L^2$-restriction theorem: For $f\in l^{2,s}(\ze^d)$ with $s>1/2$, then 
\begin{align*}
\hat{f}|_{M_{\l}}\in L_{loc}^2(M_{\l}, d\s).
\end{align*}
For $f\in l^{2,1}(\ze^d)$, we have sharper integrability of $\hat{f}|_{M_{\l}}$ near $\Sigma_{\l}$ with respect to $d\m$.
\begin{lem}
For $f\in l^{2,1}(\ze^d)$, a restriction $\hat{f}|_{M_{\l}}\in L^2_{loc}(M_{\l}, d\s)$ satisfies $\hat{f}|_{M_{\l} }\in L^2(M_{\l}, d\m)$. Moreover, we have
\begin{align}\label{reses0}
\|\hat{f}\|_{L^2(M_{\l}, d\m)}\leq C\|f\|_{l^{2,1}(\ze^d)}.
\end{align}

\end{lem}

\begin{proof}
Let $z\in \Sigma_{\l}$ and $\chi\in C^{\infty}(\T^d)$ which has a sufficiently small support near $z$. 
For proving $\hat{f}|_{M_{\l} }\in L^2(M_{\l}, d\m)$, it suffices to show
\begin{align}\label{restpf}
(\chi\hat{f})|_{M_{\l} }\in L^2(M_{\l}, d\m).
\end{align}
Moreover, we take a partition of unity $\{(\tilde{\g}_{j,a})^2\}_{j=1,...d, a=\pm}$ of $\mathbb{S}^{d-1}$ such that 
\begin{align*}
\supp \tilde{\g}_{j,a}\subset \{x\in \mathbb{S}^{d-1}\subset \re^d\mid \pm x_d\geq \frac{|x|}{2d} \}.
\end{align*}
We set $\g_{j,a}(\x)=\tilde{\g}_{j,a}((\x-z)/|\x-z|)$. 

First, for $j=1,...d$ and $a=\pm$, we shall prove
\begin{align}\label{respf1}
\int_{M_{\l}\setminus \Sigma_{\l}} |(\g_{j,a}\chi\hat{f})|_{M_{\l} }(\x)|^2d\m(\x)\leq C\int_{\re^{d-1}} \frac{|(\g_{j,a}\chi\hat{f})(\x', g(\x'))|^2}{|\x'-z'|}d\x'.
\end{align}
We may assume $j=d$ and $a=+$. We define a real-valued function $g$ by
\begin{align*}
\sin \pi g(\x')=\sqrt{(\frac{\l}{4}- \sum_{j=1}^{d-1}\sin^2\pi\x_j  )},\,\, g(\x')>0.
\end{align*}
Then $g$ satisfies
\begin{align*}
h_0(\x', g(\x'))=\l\,\, \text{for}\,\, \x=(\x', g(\x'))\in \supp (\g_{d,+}\chi)\setminus \Sigma_{\l}.
\end{align*}
We note
\begin{align}
&|\pa_{\x_d}h_0(\x', g(\x'))|\sim |g(\x')-z_d|\geq |\x'-z'|,\nonumber\\
&|\pa_{\x'}g(\x')|=|-\frac{(\pa_{\x'}h_0)(\x',g(\x'))}{(\pa_{\x_d}h_0)(\x', g(\x'))}|\sim |\frac{\x'-z'}{|g(\x')-z_d|}|\leq1.\label{respf2}
\end{align}
on $\supp (\g_{d,+}\chi)\setminus \Sigma_{\l}$. These inequalities with $(\ref{surfrep})$  implies
\begin{align*}
\int_{M_{\l}} |(\g_{j,a}\chi\hat{f})|_{M_{\l}\setminus \Sigma_{\l} }(\x)|^2d\m(\x)=&\int_{\re^{d-1}} \frac{|(\g_{j,a}\chi\hat{f})(\x', g(\x'))|^2}{|(\pa_{\x_d}h_0)(\x', g(\x'))|}d\x'\\
\leq&C\int_{\re^{d-1}} \frac{|(\g_{j,a}\chi\hat{f})(\x', g(\x'))|^2}{|\x'-z'|}d\x'.
\end{align*}

Summing $(\ref{respf1})$ over $j=1,...,d$ and $a=\pm$, we obtain
\begin{align}
\int_{M_{\l}} |(\chi\hat{f})|_{M_{\l}\setminus \Sigma_{\l} }(\x)|^2d\m(\x)\leq& C\int_{\re^{d-1}} \frac{|(\chi\hat{f})(\x', g(\x'))|^2}{|\x'-z'|}d\x' \label{reses1}\\
\leq&C\|\jap{D_{\x'}}^{1/2}((\chi\hat{f})(\x',g(\x')) )\|_{L^2(\re^{d-1})}^2\nonumber\\
\leq&C\|\chi\hat{f}\|_{H^1(\re^d)}^2.\nonumber
\end{align}
where we use the Hardy inequality in the second line and use Proposition \ref{restprop} in the third line. We recall that $\supp \chi$ is small enough and we identify the integral over $\T^d$ with the integral over this fundamental domain $[-\frac{1}{2}, \frac{1}{2}]^d$. This implies $\|\chi\hat{f}\|_{H^1(\re^d)}=\|\chi\hat{f}\|_{H^1(\T^d)}$. Since $f\in l^{2,1}(\ze^d)$, we have $\chi\hat{f}\in H^1(\T^d)$. Thus we conclude $(\ref{restpf})$. The estimate $(\ref{reses0})$ follows from $(\ref{reses1})$ by using a partition of unity and the standard $L^2$ restriction theorem.

\end{proof}

\begin{rem}
The assumption $d\geq 3$ is needed for using the Hardy inequality.
\end{rem}

Now we prove the Stone theorem near the hyperbolic threshold.

\begin{lem}
For $f\in l^{2,1}(\ze^d)$, we have
\begin{align}\label{Stone}
\frac{1}{2\pi}\Im (f,R_0(\l\pm i0)f)=\int_{M_{\l}}|\hat{f}(\x)|^2d\m(\x).
\end{align}
\end{lem}

\begin{proof}
For $\hat{f}\in C^{\infty}(\T^d)$, $(\ref{Stone})$ follows from a simple calculation. Let $f\in l^{2,1}(\ze^d)$. Take a sequence $\widehat{f_k}\in C^{\infty}(\T^d)$ such that $\widehat{f_k}\to \hat{f}$ in $H^1(\T^d)$. Then $(\ref{Stone})$ follows from $(\ref{LAP1})$ and $(\ref{reses0})$.
\end{proof}

\begin{lem}\label{vani}
Let $V$ be a real-valued function satisfying $|V|\leq C\jap{x}^{-2}$. If $u\in l^{2,-1}(\ze^d)$ satisfies $u+R_0(\l\pm i0)Vu=0$, then $\widehat{Vu}|_{M_{\l}}=0$. 
\end{lem}
\begin{proof}
We note $\widehat{Vu}|_{M_{\l}}$ and $(Vu,u)$ is well-defined, which follow from $u\in l^{2,-1}(\ze^d)$ and $Vu\in l^{2,1}(\ze^d)$. Then we have
\begin{align*}
0=-\Im(Vu, u)=\Im(Vu, R_0(\l\pm i0)Vu)=2\pi\int_{M_{\l}}|\widehat{Vu}(\x)|^2d\m(\x).
\end{align*}
Thus we obtain $\widehat{Vu}|_{M_\l}=0$.
\end{proof}

\subsection{No resonance in the interior of the spectrum}

For $0\leq k\leq d$, we define
\begin{align*}
p_k(\y)=-\sum_{j=1}^k\y_j^2+\sum_{j=k+1}^d\y_j^2.
\end{align*}

The next lemma is a weaker version of \cite[Theorem IX.41]{RS} near the hyperbolic thresholds.

\begin{lem}\label{dava}
Suppose $d\geq 3$. Let $f\in C^{1}(\T^d)$ such that $f|_{M_{\l}}=0$. Then we have $(h_0-\l)^{-1}f\in L^2(\T^d)$.
\end{lem}
\begin{rem}
We regard $(h_0-\l)^{-1}f$ as a principal-valued:
\begin{align*}
((h_0-\l)^{-1}f,\f)=\lim_{\e\to 0}\int_{|h_0-\l|>\e}\frac{f(\x)\f(\x)}{h_0(\x)-\l}d\x.
\end{align*}
However, since $f|_{M_{\l}}=0$, $(h_0-\l\pm i0)^{-1}f$ coincide with $(h_0-\l)^{-1}f$.
\end{rem}

\begin{proof}
Take $\x_0\in \T^d$ such that $h_0(\x_0)=\l$ and $dh_0(\x_0)=0$. By the Morse lemma, there exist an open neighborhood $U\subset \T^d$ and a diffeomorphism $\k$ from $U$ to its image such that $h_0(\k^{-1}(\y))-\l=p_k(\y)$ for some $0\leq k\leq d$. Set $J(\y)=|\det d\k^{-1}(\y)|$. Take a cut-off function $\chi\in C^{\infty}(\T^d, [0,1])$ such that $\supp \chi\subset U$. We only show that $\chi(h_0-\l)^{-1}f\in L^2(\T^d)$. Apart from the hyperbolic threshold, the proof is easier and omitted since $f$ vanishes at the submanifold $h_0=\l$.

We may assume that $\k(U)\subset \re^d$ is convex. We write $f_{\k}(\y)=f(\k^{-1}(\y))$ for $\y\in \supp \k(U)$.
Since $f|_{M_{\l}}=0$ holds, we have $f_{\k}(|\y''|\o_1, |\y''|\o_2)=0$, where we write $\y=(|\y'|\o_1, |\y''|\o_2)$ with $\o_1\in \mathbb{S}^{k-1}$, $\o_2\in \mathbb{S}^{d-k-1}$. 
Set 
\begin{align*}
a(\y)=\int_0^1\o_1\cdot(\pa_{\y'}f)(((1-t)|\y''|+t|\y'|)\o_1, |\y''|\o_2)dt.
\end{align*}
By Taylor expanding, we see
\begin{align*}
f_{\k}(\y)=&f_{\k}(|\y'|\o_1, |\y''|\o_2)\\
=&f_{\k}(|\y''|\o_1, |\y''|\o_2)+(|\y'|-|\y''|)\cdot a(\y)\\
=&(|\y'|-|\y''|)\cdot a(\y).
\end{align*}
Thus we have $|f_{\k}(\y)|\leq C_{\y_0}||\y'|-|\y''||$ on $\y\in \k(U)$. Hence we obtain
\begin{align*}
\int_{|\y-\y_0|\leq 1}\chi_{\k}(\y)J(\y) \frac{|f_{\k}(\y)|^2}{p_k(\y)^2}d\y\leq C_{\y_0}^2\int_{|\y-\y_0|\leq 1} \frac{1}{(|\y'|+|\y''|)^2}d\y<\infty.
\end{align*}
This implies $\chi(h_0(\x)-\l)^{-1}f\in L^2(\T^d)$.

\end{proof}

\begin{proof}[Proof of Theorem \ref{hypthm}]
By the assumption, we note $\widehat{Vu}\in C^{1}(\re^n)$ by the Sobolev embedding theorem. By Lemma \ref{vani} and Lemma \ref{dava}, we have $u\in l^2(\ze^n)$.
\end{proof}

\section{Limiting absorption principle, Proof of Theorem \ref{lapthm}}

Suppose $d\geq 3$ and $|V(x)|\leq C\jap{x}^{-2-\d}$ with $\d>0$. Fix a signature $\pm$.
Set 
\begin{align*}
\mathbb{C}_{\pm}=\{z\in \mathbb{C}\mid \pm \Im z>0\},  \,\, \overline{\mathbb{C}_{\pm}}=\{z\in \mathbb{C}\mid \pm \Im z\geq 0\}.
\end{align*}
We define $R_{0,\pm}(z)\in B(l^{2,1}(\ze^d), l^{2,-1}(\ze^d))$ for $z\in \overline{\mathbb{C}_{\pm}}$ by
\begin{align*}
R_{0,\pm}(z)=\begin{cases}(H_0-z)^{-1}\,\, \text{for}\,\, \pm\Im z>0,\\
(H_0-z\mp i0)^{-1}\,\, \text{for}\,\, z\in \re.
\end{cases}
\end{align*}
We recall from \cite[Theorem 1.8]{TT} that
\begin{align}\label{Ho}
z\in \mathbb{C}\mapsto R_{0,\pm}(z)\in B(l^{2,s}(\ze^d), l^{2,-s}(\ze^d))\,\, \text{is H\"older continuous}
\end{align}
for $s>1$.

\begin{lem}\label{Rescpt}
Let $1\leq s<1+\d$.
 Then it follows that $R_{0,\pm}(z)V$ is a compact operator in $B(l^{2,-s}(\ze^d))$ for $z\in \overline{\mathbb{C}_{\pm}}$. Moreover, a map $z\in \overline{\mathbb{C}_{\pm}}\mapsto R_{0,\pm}(z)V\in B(l^{2,-s}(\ze^d))$ is continuous.
\end{lem}
\begin{proof}
In order to prove that $R_{0,\pm}(z)V$ is compact in $B(l^{2,-s}(\ze^d))$, it suffices to prove that $\jap{x}^{-1}R_{0,\pm}(z)V\jap{x}$ is compact in $B(l^{2}(\ze^d))$. We write
\begin{align*}
\jap{x}^{-s}R_{0,\pm}(z)V\jap{x}^s=\jap{x}^{-s}R_{0,\pm}(z)\jap{x}^{-1} \times V\jap{x}^{1+s}
\end{align*}
From $(\ref{LAP1})$, we have $\jap{x}^{-s}R_{0,\pm}(z)\jap{x}^{-1}\in B(l^2(\ze^d))$. Moreover, $|V(x)|\leq C\jap{x}^{-2-\d}$ with $\d>0$ implies that $V\jap{x}^{1+s}$ is a compact operator since each multiplication operator which vanishes at infinity is a compact operator on $l^2(\ze^d)$. Thus the compactness of $R_{0,\pm}(z)V$ follows.

Next, we prove that a map $z\in \overline{\mathbb{C}_{\pm}}\mapsto \jap{x}^{-s}R_{0,\pm}(z)\jap{x}^{-1-\d}\in B(l^{2}(\ze^d))$ is continuous, which implies the continuity of $R_{0,\pm}(z)V\in B(l^{2,-s}(\ze^d))$. We may assume $\d>0$ is small enough. By $(\ref{LAP1})$ and a density argument, we have
\begin{align}\label{LAPiso}
\sup_{z\in \overline{\mathbb{C}_{\pm}}}\|\jap{x}^{-1+\d}R_{0,\pm}(z)\jap{x}^{-1-\d}\|_{B(l^2(\ze^d))}<\infty.
\end{align}
for $\d>0$ small enough. From $(\ref{LAPiso})$, we see that there exists $M>0$ such that
\begin{align}\label{3e1}
\sup_{z\in \overline{\mathbb{C}_{\pm}}}\|\jap{x}^{-s}R_{0,\pm}(z)\jap{x}^{-1-\d}\|_{B(l^2(\ze^d), l^2(|x|\geq M) )}<\frac{\e}{3}.
\end{align}
On the other hand, $(\ref{Ho})$ implies that a map
\begin{align*}
z\in \overline{\mathbb{C}_{\pm}} \mapsto \chi_{\{|x|<M\}}\jap{x}^{-s}R_{0,\pm}(z)\jap{x}^{-1-\d}\in B(l^2(\ze^d))
\end{align*}
is continuous, where $\chi_{A}$ is the characteristic function of $A\subset \re^d$. Thus there exists $\d_1>0$ such that $|z-z'|<\d_1$ with $z,z'\in \overline{\mathbb{C}_{\pm}}$ implies
\begin{align*}
\|\chi_{\{|x|<M\}}\jap{x}^{-s}R_{0,\pm}(z')\jap{x}^{-1-\d}-\chi_{\{|x|<M\}}\jap{x}^{-s}R_{0,\pm}(z')\jap{x}^{-1-\d}\|_{B(l^2(\ze^d))}<\frac{\e}{3}.
\end{align*}
This inequality with $(\ref{3e1})$ gives
\begin{align*}
\|\jap{x}^{-s}R_{0,\pm}(z')\jap{x}^{-1-\d}-\jap{x}^{-s}R_{0,\pm}(z')\jap{x}^{-1-\d}\|_{B(l^2(\ze^d))}<\e
\end{align*}
for $|z-z'|<\d$. This completes the proof.
\end{proof}

\begin{lem}\label{pres}
Let $z\in \mathbb{C}\setminus \re$ and let $s\in \re$. Then $H_0$, $H$, $(H_0-z)^{-1}$ and $(H-z)^{-1}$ preserve $l^{2,s}(\ze^d)$. In particular, $H_0-z$ and $H-z$ are invertible on $l^{2,s}(\ze^d)$.
\end{lem}

\begin{proof}
By using relations $[V,\jap{x}^s]=0$ and
\begin{align*}
[(P-z)^{-1}, \jap{x}^s]=(P-z)^{-1}[\jap{x}^s, P](P-z)^{-1},\,\, P\in \{H_0,H\},
\end{align*}
it suffices to prove $[H_0,\jap{x}^s]\jap{x}^{-s}\in B(l^2(\ze^d))$. This is easily proved since its Fourier conjugate $[h_0, \jap{D_{\x}}^s]\jap{D_{\x}}^{-s}$ of $[H_0,\jap{x}^s]\jap{x}^{-s}$ is a pseudodifferential operator of order $-1$ on $\T^d$. This completes the proof.

\end{proof}

\begin{lem}\label{hypreg}
Let $z\in \overline{\mathbb{C}_{\pm}}$. Suppose that $u\in l^{2,-1-\d}(\ze^d)$ satisfies $(I+R_{0,\pm}(z)V)u=0$. Then we have $u\in l^{2,-1}(\ze^d)$.
\end{lem}
\begin{proof}
This lemma immediately follows from $|V|\leq C\jap{x}^{-2-\d}$ and $(\ref{LAP1})$.
\end{proof}

\begin{prop}\label{lappropp}
Let $U\subset \mathbb{C}_{\pm}$ be a bounded open set satisfying
\begin{align}\label{Ker}
\{u\in l^{2,-1}(\ze^d)\mid (I+R_{0,\pm}(z)V)u=0\}=\{0\},\,\, \text{for any}\,\, z\in \overline{U}.
\end{align}
\item[$(i)$] Let $1\leq s<1+\d$. Then an inverse $(I+R_{0,\pm}(z)V)^{-1}\in B(l^{2,-s}(\ze^d))$ exists for $z\in \overline{U}$ and
\begin{align*}
\sup_{z\in \overline{U}}\|(I+R_{0,\pm}(z)V)^{-1}\|_{B(l^{2,-s}(\ze^d))}<\infty.
\end{align*}

\item[$(ii)$] For $z\in \overline{U}$, we set
\begin{align*}
R_{\pm}(z)= (I+R_{0,\pm}(z)V)^{-1}R_{0,\pm}(z)\in B(l^{2,1}(\ze^d), l^{2,-1}(\ze^d)).
\end{align*}
Then we have $R_{\pm}(z)=(H-z)^{-1}$ for $z\in \overline{U}\setminus \re$ and
\begin{align*}
\sup_{z\in \overline{U}}\|R_{\pm}(z)\|_{B(l^{2,1}(\ze^d), l^{2,-1}(\ze^d))}<\infty.
\end{align*}

\item[$(iii)$] Let $1<s\leq 1+\d/2$. Then a map $z\in \overline{U}\mapsto R_{\pm}(z)\in B(l^{2,s}(\ze^d), l^{2,-s}(\ze^d))$ is H\"older continuous.

\end{prop}

\begin{proof}
Lemma \ref{Rescpt} implies that $\{I+R_{0,\pm}(z)V\}_{z\in \overline{U}}$ is a continuous family of Fredholm operators with index $0$ on $B(l^{2,-s}(\ze^d))$. Thus the assumption $(\ref{Ker})$ implies that $I+R_{0,\pm}(z)V$ is invertible for $z\in \overline{U}$ and that a map $z\mapsto (I+R_{0,\pm}(z)V)^{-1}\in B(l^{2,-s}(\ze^d))$ is continuous. This with the compactness of $\overline{U}$ gives the proof of $(i)$.

The part $(ii)$ follows from the part $(i)$, $(\ref{LAP1})$ and the resolvent equation:
\begin{align*}
(I+(H_0-z)^{-1}V)(H-z)^{-1}=(H_0-z)^{-1},\,\,z\in \mathbb{C}\setminus \re.
\end{align*}

To prove part $(iii)$, we observe that $z\in \overline{U}\mapsto (I+R_{0,\pm}(z)V)^{-1}\in B(l^{2,-s}(\ze^d))$ is H\"older continuous. In fact, for $z, z'\in \overline{U}$, we have
\begin{align*}
&(I+R_{0,\pm}(z)V)^{-1}-(I+R_{0,\pm}(z')V)^{-1}\\
&=(I+R_{0,\pm}(z)V)^{-1}(R_{0,\pm}(z')-R_{0,\pm}(z))V(I+R_{0,\pm}(z')V)^{-1}.
\end{align*}
Part $(i)$, $(\ref{Ho})$, and $V\in B(l^{2,-s}(\ze^d), l^{2,s}(\ze^d))$ imply the H\"older continuity of $(I+R_{0,\pm}(z)V)^{-1}$. This, $(\ref{Ho})$ and the following representation:
\begin{align*}
R_{\pm}(z)-R_{\pm}(z')=&(I+R_{0,\pm}(z)V)^{-1}(R_{0,\pm}(z)-R_{0,\pm}(z'))\\
&+ ((I+R_{0,\pm}(z)V)^{-1}-(I+R_{0,\pm}(z')V)^{-1})R_{0,\pm}(z'),
\end{align*}
finish the proof of part $(iii)$.

\end{proof}

\begin{proof}[Proof of Theorem \ref{lapthm}]
From now on, we assume that $V$ is a finitely supported potential. We take $R>0$ such that $\s(H)\subset \{|z|<R\}$.
Then $(\ref{Ker})$ holds for 
\begin{align*}
U=\{z\in \mathbb{C}\mid \pm\Im z\geq 0,\,\,  |z|<2R,\,\, |z|>\e_1,\,\, |z-4d|>\e_1\}.
\end{align*}
Moreover, we note $\s(H)\cap \O_{\e_1,\pm}\setminus U=\emptyset$. Now Theorem \ref{lapthm} follows from Corollary \ref{corab} and Proposition \ref{lappropp}.

\end{proof}

\appendix

\section{Lorentz space}

For a measure space $(X,\m)$, $L^{p,r}(X,\m)$ denotes the Lorentz space for $1\leq p\leq \infty$ and $1\leq r\leq \infty$:
\begin{align*}
&\|f\|_{L^{p,r}(X)}=\begin{cases}
p^{\frac{1}{r}}(\int_0^{\infty}\m(\{x\in X\,|\, |f(x)|>\a\})^{\frac{r}{p}} \a^{r-1}d\a)^{\frac{1}{r}},\quad &r<\infty,\\
\sup_{\a>0}\a\m(\{x\in X\mid |f(x)|>\a\})^{\frac{1}{p}},\quad &r=\infty,
\end{cases}\\
&L^{p,r}(X,\m)=\{f:X\to \mathbb{C}\mid f:\text{measurable},\, \|f\|_{L^{p,r}(X)}<\infty\}.
\end{align*}
Moreover, we denote $L^{p,r}(\re^d)=L^{p,r}(\re^d, \m_L)$ and $l^{p}_r(\ze^d)=L^{p,r}(\ze^d, \m_c)$, where $\m_L$ is the Lebesgue measure on $\re^d$ and $\m_c$ is the counting measure on $\ze^d$. For a detail, see \cite{G}. In this section, we state some fundamental properties of the Lorentz spaces. Note that $L^{p,p}(X,\m)=L^{p}(X)$.

\begin{lem}[The Young inequalities in the Lorentz spaces]
Let $1<p_i<\infty$, $1\leq q_i\leq \infty$ with $\frac{1}{r}=\frac{1}{p_1}+\frac{1}{p_2}-1>0$ and $s\geq 1$ with $\frac{1}{q_1}+\frac{1}{q_2}\geq \frac{1}{s}$. Then we have
\begin{align*}
\|f\ast g\|_{l^{r}_s(\ze^d)}\leq C\|f\|_{l^{p_1}_{q_1}(\ze^d)}\|g\|_{l^{p_2}_{q_2}(\ze^d)}.
\end{align*}
\end{lem}

\begin{lem}[The H\"older inequalities in the Lorentz spaces]\label{Hol}
If $1\leq p_1,p_2,q_1,q_2\leq \infty$ and $1\leq r\leq \infty$ satisfy
\begin{align*}
\frac{1}{p_1}+\frac{1}{p_2}=\frac{1}{r}<1,
\end{align*}
then
\begin{align*}
\|fg\|_{l^{r}_{\min(q_1,q_2)}(\ze^d)}\leq \|f\|_{l^{p_1}_{q_1}(\ze^d)}\|g\|_{l^{p_2}_{q_2}(\ze^d)}.
\end{align*}

\end{lem}

\section{Harmonic analysis}

\begin{prop}\label{mlem}
Let $m\in C^{\infty}(\re^d\setminus\{0\})$ satisfying that there exists $0\leq k<d$ such that
\begin{align}\label{m}
|\pa_{\x}^{\a}m(\x)|\leq C_{\a}|\x|^{-k-|\a|},\quad x\in \re^d.
\end{align}
Moreover, we assume that $m$ is compactly supported.
Then if we set
\begin{align*}
I=\int_{\re^d}e^{-2\pi ix\cdot \x}m(\x)d\x,
\end{align*}
then $|I|\leq C\jap{x}^{-d+k}$.
\end{prop}

\begin{proof}
Since $m$ is compactly supported, we may assume $|x|\geq 1$. Take $\chi\in C_c^{\infty}(\re^d)$ such that $\chi=1$ on $|\x|\leq 1$ and $\chi=0$ on $|\x|\geq 2$. Set $\bar{\chi}=1-\chi$. For $\d>0$, we have
\begin{align*}
I=\int_{\re^d}(\chi(\x/\d)+\bar{\chi}(\x/\d))e^{-2\pi ix\cdot \x}m(\x)d\x=:I_1+I_2.
\end{align*}
We learn
\begin{align*}
|I_1|\leq \int_{|\x|\leq 2\d}|\chi(\x/\d)||\x|^{-k}d\x\leq C\d^{d-k}.
\end{align*}
By integrating by parts, for $N> d-k$ we have
\begin{align*}
|I_2|\leq C&|x|^{-N}\sum_{|\a|= N}|\int_{\re^d}e^{-2\pi ix\cdot \x} D_{\x}^{\a}(\bar{\chi}(|\x|/\d)m(\x)) d\x|\\
\leq &C|x|^{-N}\sum_{|\a|= N}|\sum_{\b\leq \a}\int_{\re^d}e^{-2\pi ix\cdot \x} D_{\x}^{\b}(\bar{\chi}(|\x|/\d))\pa_{\x}^{\a-\b}m(\x) d\x|\\
\leq&C|x|^{-N}\sum_{|\a|\leq N}\sum_{\b\leq \a}\int_{\re^d} \d^{-\b}\bar{\chi}^{|\b|}(|\x|/\d)|\x|^{-k-(N-|\b|)} d\x.
\end{align*}
For $\b=0$, 
\begin{align*}
\int_{\re^d}\bar{\chi}(\x/\d)|\x|^{-k-N}d\x\leq C\d^{d-k-N}
\end{align*}
follows and for $\b\neq 0$,
\begin{align*}
\int_{\re^d} \d^{-\b}\bar{\chi}^{|\b|}(|\x|/\d)|\x|^{-k-(N-|\b|)} d\x\leq& C\int_{\d\leq |\x|\leq 2\d}\d^{-\b}|\x|^{-k-N+|\b|}d\x\\
\leq&C\d^{d-k-N}.
\end{align*}
These imply $|I_2|\leq C|x|^{-N}\d^{d-k-N}$. We set $\d=|x|^{-1}$ and obtain $|I|\leq C|x|^{-d+k}$ for $|x|\geq 1$.
\end{proof} 

\begin{cor}\label{Kerpo}
Let $d\geq 1$, $0<l<d$ and $K_{l}$ be defined by
\begin{align*}
K_l(x)=\int_{\T^d}e^{2\pi ix\dot \x}h_0(\x)^{-l/2}d\x.
\end{align*}
Then we have a pointwise bound $|K_l(x)|\leq C\jap{x}^{-d+l}$.
\end{cor}
\begin{proof}
By the Morse lemma, we have $h_0(\x)^{-l/2}\sim |\x|^{-l}$ near $\x=0$. Moreover, it follows that $h_0(\x)^{-l/2}$ is smooth away from $\x=0$. Applying Proposition \ref{mlem}, we obtain $|K_l(x)|\leq C\jap{x}^{-d+l}$.
\end{proof}

Now we define operators $H_0^{-l/2}$ for $0<l<d$ by
\begin{align*}
H_0^{-l/2}u(x)=\sum_{y\in \ze^d}K_l(x-y)u(y),\,\, u\in \bigcap_{s>0}l^{2,s}(\ze^d).
\end{align*}
It is easily seen that $H_0^{-l/2}$ is a continuous linear operator:
\begin{align*}
H_0^{-l/2}:\bigcap_{s>0}l^{2,s}(\ze^d)\to \bigcup_{s\in \re}l^{2,s}(\ze^d).
\end{align*}

The next corollary implies that $H_0^{-1}$ can be uniquely extended to the continuous linear operator from $l^{2,\a}(\ze^d)$ to $l^{2,-\b}(\ze^d)$ for $\a,\b>1/2$ with $\a+\b\geq 2$.

\begin{cor}[Discrete version of the HLS inequality]\label{disHLS}
Let $d\geq 1$ and $0<k<d$. Then $H_0^{-l/2}$ is bounded from $l^{p}_r(\ze^d)$ to $l^{q}_r(\ze^d)$ if $1< p< q < \infty$ satisfies
\begin{align}\label{HLSindex}
\frac{1}{p}-\frac{1}{q}=\frac{l}{d}
\end{align}
and $1\leq r\leq \infty$. 

Moreover, if $W_1\in l^{r_1}_{\infty}(\ze^d)$ and $W_2\in l^{r_2}_{\infty}(\ze^d)$ with $1/r_1+1/r_2=l/d$ with $r_1,r_2>2$.
Then we have 
\begin{align*}
W_1H_0^{-l/2}W_2\in B(l^2(\ze^d))
\end{align*}

 In particular, $\jap{x}^{-\a}H_0^{-1}\jap{x}^{-\b}\in B(l^2(\ze^d))$ if $\a+\b\geq 2$ and $\a,\b>0$ if $d\geq 4$ and $\a+\b\geq 2$ and $\a,\b>1/2$ if $d=3$.
\end{cor}

\begin{rem}\label{Hardy}
This corollary gives $H_0^{-l/2}\jap{x}^{-l}\in B(l^2(\ze^d))$ for $0<l<d$. In fact,
\begin{align*}
\|H_0^{-l/2}\jap{x}^{-l}f\|_{l^2(\ze^d)}\leq C\|H_0^{-l/2}\|_{B(l^{\frac{2l}{l+2d}}_{\infty}(\ze^d), l^2(\ze^d))} \|\jap{x}^{-l}\|_{l^{\frac{d}{l}}_{\infty}(\ze^d)} \|f\|_{l^2(\ze^d)}.
\end{align*}
These are exactly the discrete Hardy inequalities.
\end{rem}

\section{Restriction theorem for a Lipschitz manifold}

In this appendix, we prove the $L^2$-restriction theorem for a Lipschitz manifold. The proof is standard, however, we give its proof for readers' convenience.

\begin{lem}\label{restapplem}
Let $f\in H^{1}(\re^d)$ and $g$ be a real-valued Lipschitz function on $\re^{d-1}$. Then it follows that $k(\x)=f(\x', \x_d+g(\x'))$ belongs to $H^1(\re^d)$ and there exists $C>0$ which is independent of $f$ such that
\begin{align}\label{restapp}
\|k\|_{H^1(\re^d)}\leq C\|f\|_{H^1(\re^d)}.
\end{align}
\end{lem}

\begin{proof}
By changing of variables, we have $\|k\|_{L^2(\re^d)}=\|f\|_{L^2(\re^d)}$. For $j=1,...d-1$, we have
\begin{align*}
\pa_{\x_j}(k(\x',\x_d+g(\x')))=&(\pa_{\x_j}k)(\x',\x_d+g(\x'))+(\pa_{\x_j}g)(\x')(\pa_{\x_d}k)(\x',\x_d+g(\x')),\\
\pa_{\x_d}(k(\x',\x_d+g(\x')))=&(\pa_{\x_d}k)(\x'+\x_d+g(\x')).
\end{align*}
Using this computation, we obtain (\ref{restapp}).
\end{proof}

\begin{prop}\label{restprop}
Under the assumption of Lemma \ref{restapplem}, we have
\begin{align*}
\|\jap{D_{\x'}}^{1/2}(f(\x', g(\x')))\|_{L^2(\re^{d-1})}\leq C\|f\|_{H^1(\re^d)}.
\end{align*}
\end{prop}
\begin{proof}
In the following, we denote the Fourier transform of $f$ by $\hat{f}$.
By changing of variable and and the H\"older ineqality, we have
\begin{align*}
|\int_{\re^{d-1}}f(\x',g(\x'))e^{-2\pi ix'\cdot \x'}d\x'|=&|\int_{\re^d} \hat{k}(x)dx_d|\\
\leq& (\int_{\re}\jap{x}^{-2}dx_d)^{1/2} (\int_{\re}|\jap{x}\hat{k}(x)|^2dx_d)^{1/2}\\
\leq&C\jap{x'}^{-1/2}(\int_{\re}|\jap{x}\hat{k}(x)|^2dx_d)^{1/2}.
\end{align*}
Thus we have
\begin{align*}
\|\jap{D_{\x'}}^{1/2}(f(\x', g(\x')))\|_{L^2(\re^{d-1})}^2=&\|\jap{x'}^{1/2}\widehat{f(\x',g(\x'))}(x)\|_{L^2(\re^{d-1})}^2\\
\leq&C^2\|\jap{x}\hat{k}\|_{L^2(\re^d)}^2\\
=&C^2\|k\|_{H^1(\re^d)}^2.
\end{align*}
This computation with Lemma \ref{restapplem} completes the proof.

\end{proof}

\end{document}